\documentclass[11pt,reqno]{amsart}
\usepackage{amsfonts, amssymb, amsmath, enumerate, bm, xspace, graphicx, caption, graphics}
\usepackage[linktoc=page, colorlinks, linkcolor=blue, citecolor=blue]{hyperref} 
\usepackage{pbox}
\numberwithin{equation}{section}

\DeclareMathOperator{\E}{\mathbb{E}}

\DeclareMathOperator{\sign}{sign}

\newcommand{\ip}[2]{\langle#1,#2\rangle}

\def \R {\mathbb{R}}

\def \NN {\mathcal{N}}

\def \d {\delta}

\def \tran {\mathsf{T}}

\def \< {\langle}
\def \> {\rangle}

\newtheorem{theorem}{Theorem}[section]
\newtheorem{proposition}[theorem]{Proposition}
\newtheorem{corollary}[theorem]{Corollary}
\newtheorem{lemma}[theorem]{Lemma}

\newtheorem{definition}[theorem]{Definition}

\theoremstyle{remark}
\newtheorem{remark}[theorem]{Remark}

\newcommand{\vect}[1]{\bm{#1}}
\newcommand{\mat}[1]{\bm{#1}}

\def \g {\vect{g}}
\def \x {\vect{x}}
\def \a {\vect{a}}
\def \z {\vect{z}}
\def \h {\vect{h}}
\def \u {\vect{u}}
\def \v {\vect{v}}
\def \xhat {\widehat{\vect{x}}}

\def \y {\vect{y}}
\def \A {\mat{A}}
\def \G {\mat{G}}
\def \Id {\mat{I}}
\def \P {\mat{P}}

\title[]{The generalized Lasso with non-linear observations}

\author{Yaniv Plan \and Roman Vershynin}

\address{Y.~Plan is with the Department of Mathematics, University of British Columbia 1984 Mathematics Rd., Vancouver, BC V6T 1Z2, Canada}
\email{yaniv@math.ubc.ca}
\address{R.~Vershynin is with the Department of Mathematics, University of Michigan, 530 Church St., Ann Arbor, MI 48109, U.S.A.}
\email{romanv@umich.edu}

\thanks{R. V. is partially supported by NSF grant 1265782 and U.S. Air Force grant FA9550-14-1-0009.}

\date{\today}

\begin{document}

\begin{abstract}  
We study the problem of signal estimation from non-linear observations when the signal belongs to a low-dimensional set buried in a high-dimensional space.  
A rough heuristic often used in practice postulates that \textit{non-linear} observations may 
be treated as {\em noisy linear} observations, and thus the signal may be estimated using the generalized Lasso.  This is appealing because of the abundance of efficient, specialized solvers for this program.  
Just as noise may be diminished by projecting onto the lower dimensional space, the error from modeling non-linear observations with linear observations will be greatly reduced when using the signal structure in the reconstruction.  We allow general signal structure, only assuming that the signal belongs to some  set $K \subset \R^n$.  We consider the single-index model of non-linearity.  Our theory allows the non-linearity to be discontinuous, not one-to-one
and even unknown. 
We assume a random Gaussian model for the measurement matrix, but allow the rows to have an unknown covariance matrix.  As special cases of our results, we recover near-optimal theory for noisy linear observations, and also give the first theoretical accuracy guarantee for 1-bit compressed sensing with unknown covariance matrix of the measurement vectors.  
\end{abstract}

\maketitle

\section{Introduction}

Before describing to the non-linear setting which is the main theme of this paper, let us first consider the structured linear model
\[ \y = \A \x + \z\]
where an unknown vector $\x$ belongs to some known set $K \subset \R^n$.  
The goal is to reconstruct the signal $\x$ from the noisy measurement vector $\y \in \R^m$.  
A common method is to minimize the $\ell_2$ loss subject to a structural constraint:
\begin{equation}         \label{eq: estimation}
\text{minimize } \|\A \x' - \y\|_2 \quad \text{ subject to }\quad \x' \in K.
\end{equation}
We shall refer to this generalized Lasso as the {\em $K$-Lasso} for the rest of the paper.  
The set $K$ is meant to capture structure of the signal. In many cases of interest
$K$ behaves as if it were a {\em low-dimensional} set, although it often has full linear algebraic dimension.
For example, to promote sparsity of the solution, one can choose $K$ to be a scaled $\ell_1$ ball, 
and this gives the vanilla Lasso as proposed by R.~Tibshirani \cite{Lasso}.
When the signals are matrices, to promote low rank one can choose $K$ 
to be a scaled ball in the nuclear norm, and this is referred to as the matrix Lasso \cite{tight_low_rank} or trace Lasso \cite{grave2011trace}.

How well can the signal be reconstructed based on the complexity of the set $K$?  Under the linear model, the last two decades have seen the development of a strong theoretical backing for the Lasso from the statistical community, mostly based on a sparsity assumption.  See, e.g., \cite{bunea2007aggregation,greenshtein2006best, bickel2009simultaneous, lockhart2014significance, negahban2012unified, van2009conditions, candes2009near}.  Further, recent results developed from the compressed sensing community give a clean, comprehensive theory for arbitrary signal structure.  See Section \ref{sec: related literature}.  

Consider the more challenging situation, in which there is an unknown non-linearity in the observations. We ask:  
\begin{quote}
\textit{What happens when the $K$-Lasso is used to reconstruct a signal based on non-linear observations?} 
\end{quote} 
On the one hand, Lasso is by design a method for linear regression, and it is dubious
to expect it to work if $\y$ depends non-linearly on $\A \x$.
On the other hand, practitioners have been successfully using Lasso for non-linear (especially binary) 
observations without theoretical backing.

In this paper we demonstrate that $K$-Lasso can be used for non-linear observations. 
We will see that from Lasso's point of view, {\em non-linear} observations behave as scaled and 
{\em noisy linear} observations, and we will characterize the scaling and the noise.   
Furthermore, we assume $\A$ to be Gaussian, but in contrast to much of the literature, 
we allow {\em unknown covariance} of rows.  A particular non-linearity of interest in signal processing is 1-bit quantization, which, when combined with sparse signal structure, leads to the model of 1-bit compressed sensing.  
We believe all previous theoretical results in this area have required knowledge of the covariance of rows 
for the recovery algorithm to be accurate; our work broadens the theory by removing this requirement.
We will describe related literature regarding non-linear observations in Section \ref{sec: related literature} below.

\subsection{Model}		\label{sec: model}

We will work with semiparametric single-index model of a similar form to the one in \cite{pvy}.
Let $\x \in K \subset \R^n$ be a fixed (unknown) \emph{signal} vector, let $\a_i \sim \NN(0, \Sigma)$ be independent random measurement vectors, and let $\A$ be the matrix whose $i$-th row is $\a_i^\tran$.  Let $f_i : \R \to \R$ be independent copies of an \textit{unknown}, random function $f$ 
modeling the non-linearity (it also may be deterministic), which are independent of $\A$. 
We assume that the $m$ observations $y_i$ that form the vector $\y = (y_1,\ldots,y_m)$ 
take the form 
\begin{equation}         \label{eq: model}
y_i = f_i(\ip{\a_i}{\x}).
\end{equation}
Note that the norm of $\x$ is sacrificed in this model since it may be absorbed into the unknown random function $f_i$.  Thus, to simplify presentation, we will assume that $\|\sqrt{\Sigma} \x\|_2 = 1$. We will remark on how to remove this assumption by a rescaling argument.

\subsection{Examples}

We now give two concrete examples of the above model: quantized and binary observations.

A first non-linearity of interest is quantization applied to linear observations.  Then the function $f$ maps $\ip{a_i}{x}$ to a finite alphabet of real numbers.  In this case, the non-linearity is known, and furthermore, it is designed.  Thus, the theoretical error bounds we develop below may be tuned to optimize the error.  This observation has been made in \cite{thrampoulidis2015lasso}.  

On the extreme end, one may consider {\em 1-bit quantization:} $f(\ip{a_i}{x}) = \text{sign}(\ip{a_i}{x})$.  
Measurements of this kind are of special interest due to the simplicity of hardware implementation, and the robustness to multiplicative errors.  We further discuss 1-bit quantization in Section \ref{sec: 1-bit} below.

Interestingly, binary statistical models are quite similar.  For example, $f(\ip{a_i}{x}) = \text{sign}(\ip{a_i}{x} + z_i)$ gives the {\em logistic regression} model, provided that $z_i$ is logit noise.  Other binary models are available by adjusting the distribution of $z_i$.  The classical approach in these models is (regularized) maximum likelihood estimation \cite{negahban2012unified, 1-bit_MC}.  However, it requires knowledge of the form of the nonlinearity, which is equivalent to knowledge of the distribution of $z_i$, and in practice one would often not expect this to be known.  Further, the theory requires the log-likelihood to be \textit{strongly convex}, which ceases to hold when $z_i$ is small compared to $||x||_2$.  Ironically, the noise needs to be roughly larger than the signal in the theoretical treatment of maximum-likelihood estimation (see \cite{1-bit_MC} for a discussion of this point).  In contrast, as we show, the $K$-Lasso does not need knowledge of the non-linearity, and is accurate even when the noise $z_i$ disappears, as in the 1-bit compressed sensing model.

\subsection{Simplified results when $K$ is a subspace} 

To begin in a simpler setting, let us assume that the covariance matrix $\Sigma$ is identity, 
$K$ is a $d$-dimensional subspace, and there is no non-linearity, just an unknown rescaling and noise.  
Thus, we assume that $f_i(u)= \mu u + z_i$  for $u \in \R$,
where $\mu > 0$ and $z_i \sim \NN(0, \sigma^2)$.  Then the observations take the form
\begin{equation}         \label{eq: noisy linear model}
y_i = \mu \ip{\a_i}{\x} + z_i.
\end{equation}
The $K$-Lasso \eqref{eq: estimation} becomes the \textit{least squares estimator} whose behavior is well known.  Let $\xhat$ be the solution to the $K$-Lasso.  Then, the conditional expectation of the squared error with respect to $\A$ satisfies
\[\E \| \xhat - \mu \x\|_2^2 = \sigma^2 \cdot \sum_{i=1}^d \frac{1}{\sigma_i^2(\A_K)}\] 
where $\sigma_i(\A_K)$ is the $i$-th singular value of $A$ restricted to the subspace $K$.  
Since $\A$ is Gaussian, it is well conditioned with high probability as long as the number of observations $m$ is significantly larger than the dimension $d$ of $K$ \cite{V}.  In this case, with high probability, each singular value does not deviate significantly from $\sqrt{m}$ \cite{V} and thus 
\[\E \| \xhat - \mu \x\|_2^2 \approx \frac{d}{m} \sigma^2.\]

Let us make a few observations about the ingredients involved in the above calculation.  First, the $K$-Lasso gives an estimate of a scaled version of $\x$.  Second, note the vital requirement that the number of observations $m$ exceeds the dimension of the subspace $d$.
Third, observe that the size of the scaling and the noise satisfy
\[\mu = \E(f(g) \cdot g) \qquad \text{and} \qquad \sigma^2 = \E(f(g) - \mu g)^2 = \E f(g)^2 - \mu^2,\] 
where $g$ is a standard normal random variable.  

Our main result states that up to a small extra summand, the $K$-Lasso gives the same accuracy for 
{\em non-linear} observations, with $\sigma$ and $\mu$ measured in the same way. To easily compare, we first state this result when $K$ is a subspace. Here and in the rest of the paper, a statement is said to hold with high probability if it holds with probability at least $0.99$.  Further, the symbol $\lesssim$ hides an absolute constant.

\begin{proposition}[Non-linear estimation on a subspace]
\label{prop: subspace}
  Suppose that $\a_i \sim \NN(0, \Id)$, and that $\y$ follows the semi-parametric single index model of Section \ref{sec: model}.  Let $K$ be a $d$-dimensional subspace and assume $\x \in K \cap S^{n-1}$.  
  Suppose that
  \[m \gtrsim d.\]
   Then, with high probability, the non-linear estimator $\xhat$ which minimizes the $K$-Lasso \eqref{eq: estimation} satisfies 
  \begin{equation}         \label{eq: subspace}
  \|\xhat - \mu \x\|_2 \lesssim   \frac{\sqrt{d} \, \sigma + \eta}{\sqrt{m}} 
  \end{equation}
  where 
  \begin{equation}\label{eq: parameters}
  \mu := \E[f(g) \cdot g], \qquad \sigma^2 := \E (f(g) - \mu g)^2, \qquad \eta^2 := \E (f(g) - \mu g)^2 g^2.
  \end{equation}
\end{proposition}

One sees that this mirrors the result for linear observations aside from the extra summand $\eta/\sqrt{m}$, 
which becomes quite small with a moderate number of observations $m$.
For example, in the noisy linear model \eqref{eq: noisy linear model} one has $\eta = \sigma$, 
so this result gives the classic error rate as a special case.  

Results of the above flavour have been rigorously proven in the statistics literature \cite{brillinger2012generalized}, with a focus on asymptotic behaviour of the error.  In this paper, we extend these ideas to modern trends in signal processing and statistics, in which it is assumed that the signal belongs to some non-linear low-dimensional signal structure, such as the set of sparse vectors or low-rank matrices.
 We now proceed to our main results in which $K$ will be allowed to be a general set.
\subsection{Main results}  
We will give two results below, one specialized to the case when 
the scaled signal $\mu \x$ lies at an extreme point of $K$ with (small) \textit{tangent cone}, 
and one which only assumes that $\mu \x$ lies in $K$.

\begin{definition}[Tangent cone]
The tangent cone\footnote{To allow non-convex K, the above is slight variation on the standard definition of tangent cone \cite{jahn2007introduction}.  The tangent cone may also be called the \textit{descent cone}.} of $K$ at $\x$ is
\[D(K, \x) := \{\tau \h : \tau \geq 0, \h \in K - \x\}.\]
\end{definition}
For sets with non-smooth boundary, such as the $\ell_1$ ball or the nuclear norm ball, 
the tangent cone at a boundary point can be quite narrow, and intuitively should
behave like a low-dimensional subspace.  We give an illustrative example of a tangent cone in Figure \ref{fig: tangent cone}, although, in a two-dimensional representation, we cannot do justice to the high-dimensional effects which allow convex sets to have extremely narrow tangent cones.  

\begin{figure}
\centering
\includegraphics[width = .4\textwidth]{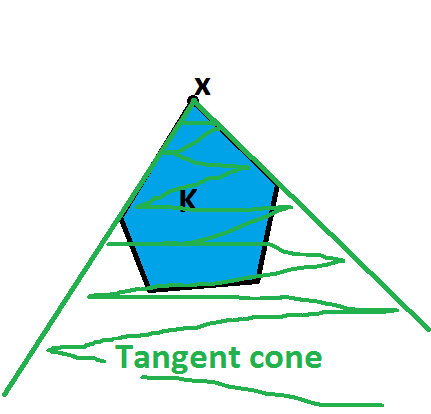}
\caption[]{The tangent cone}
\label{fig: tangent cone}
\end{figure}


While $\A$ may be singular, it can be quite well conditioned when restricted to the tangent cone; it is not surprising that this restricted conditioning of $\A$ can determine the accuracy of the solution to \eqref{eq: estimation}.  Further, this restricted condition number can be well understood via Gordon's escape through the mesh theorem (see Theorem \ref{thm: gordon}). It states that the restriction of $\A$ onto $K$ is well conditioned provided that the number of observations $m$ exceeds the {\em effective dimension} of $K$. 
The effective dimension is measured in Gordon's theorem by the notion of Gaussian mean width. 
Let us recall the notion of the local (Gaussian) mean width; see \cite{PV_IEEE, pvy, V} for further discussion 
of the mean width and how is serves as a measure of effective dimension.
\begin{definition}[Local mean width]
  The {\em local mean width} of a subset $K \subset \R^n$ is a function of scale $t \ge 0$
  defined as
  $$
  w_t(K) = \E \sup_{\x \in K \cap tB_2} \ip{\x}{\g},
  $$
  where $B_2$ denotes the unit Euclidean ball in $\R^n$.
\end{definition}

Let us pause to explain the heuristic meaning of the local mean width of a cone $D$. 
The square of the mean width, $w_1(D)^2$, can be described as a measure of {\em effective dimension} 
of $D$. This can be seen on the following two examples. 
First, let $D$ be $d$-dimensional subspace in $\R^n$. It is not difficult to check that 
$$
w_1(D)^2 \sim d,
$$
up to a absolute multiplicative constants. Thus in this case, the square of the mean width 
is equivalent to the algebraic dimension $d$.

A deeper example is where $D = D(B_1^n,\x)$ is the descent cone of the unit $\ell_1$
ball $B_1^n = \{ \u \in \R^n :\; \|\u\|_1 \le 1\}$ at some point $\x$ on the boundary of $B_1^n$.
Suppose $\x$ is $s$-sparse, meaning that $\x$ has $s$ non-zero coordinates. It should be clear 
that the smaller sparsity $s$, the thinner the descent cone $D$ is. Quantitatively, this is captured
by the notion of local mean width, which can be shown (see e.g. \cite{chandrasekaran2012convex})
to behave as follows:
$$
w_1(D)^2 \sim s \log(n/s).
$$
Thus, up to a logarithmic factor, the square of the mean width is again equivalent to the dimensionality 
of the signal $\x$, which is its sparsity $s$. 

We refer the reader to \cite[Section 2]{PV_IEEE} where the notion of mean 
width is discussed in more detail, as well as to \cite{ALMT} where an equivalent concept 
of {\em statistical dimension} is introduced. 

\medskip

Let us first state our first main result specialized to the case when $\Sigma = \Id$ and to descent-cone structure. 

\begin{theorem}[Non-linear estimation with tangent cone structure]				\label{thm: main}
  Suppose that $\a_i \sim \NN(0, \Id)$, 
  $\x \in S^{n-1}$, 
  and that $\y$ follows the semi-parametric single index model of Section \ref{sec: model}.
  Assume that $\mu \x \in K$, and let 
  $d(K): = w_1(D(K, \mu \x))^2$.
  Suppose that
  \[m \gtrsim d(K).\]
  Then, with high probability, the solution $\xhat$ of the $K$-Lasso \eqref{eq: estimation} satisfies 
  \begin{equation}         \label{eq: main}
  \|\xhat - \mu \x\|_2 \lesssim \frac{\sqrt{d(K)} \, \sigma + \eta}{\sqrt{m}}  
  \end{equation}
  where $\mu$, $\eta$, and $\sigma$ are defined in \eqref{eq: parameters}.
\end{theorem}

It should be clear that this result extends Proposition \ref{prop: subspace}
from linear to non-linear observations, and from subspaces to general sets. 
To see this, recall our observation that if $K$ is a $d$-dimensional subspace, then 
$d(K) \sim d$ up to an absolute constant factor.


\begin{remark}[Boundary of $K$]
For the above theorem to be especially useful, $\mu \x$ needs to lie on the boundary 
of $K$. Otherwise, the tangent cone is the entire $\R^n$, and the effective dimension $d(K)$ is 
of order of $n$.
In this case, the estimate becomes accurate only when the number of observations $m$ exceeds the 
ambient dimension $n$ rather than the effective dimension of the cone, which may be significantly 
smaller.  Thus, in practice, one would like to rescale $K$ to put $\mu \x$ on the boundary. 
 If $\mu \x$ does not lie precisely on the boundary, we may appeal to our more general Theorem \ref{thm: main no cone} below.  Further, we note that the unconstrained version of the $K$-Lasso overcomes this obstacle.  This has been proven in the asymptotic setting in \cite{thrampoulidis2015lasso}, which built upon the ideas in this papper.
\end{remark}

A substitution argument generalizes the above result to allow an unknown covariance matrix.

\begin{corollary}[Non-linear estimation with unknown covariance matrix]				\label{cor: unknown covariance}
\label{cor: main}
  Suppose that $\a_i \sim \NN(0, \Sigma)$, 
  $\sqrt{\Sigma} \x \in S^{n-1}$, and that $\y$ follows the semi-parametric single index model of Section \ref{sec: model}. 
  Assume that $\mu \x \in K$, and let $d(K, \Sigma): = w_1(\sqrt{\Sigma} D(K, \mu \x))^2$.
  Suppose that
  \[m \gtrsim d(K, \Sigma).\]
   Then, with high probability, the non-linear estimator $\xhat$ which minimizes the $K$-Lasso \eqref{eq: estimation} satisfies 
  \begin{equation}         \label{eq: main with Sigma}
  \|\sqrt{\Sigma}(\xhat - \mu \x)\|_2 
    \lesssim \frac{\sqrt{d(K, \Sigma)} \, \sigma + \eta}{\sqrt{m}}  
  \end{equation}
  where $\mu$, $\eta$, and $\sigma$ are defined in \eqref{eq: parameters}.
\end{corollary}

\begin{proof}
We may set $\a_i := \sqrt{\Sigma}\g_i $ where $\g_i \sim \NN(0, \Id)$.  Then $\ip{\a_i}{\x} = \ip{\g_i}{\sqrt{\Sigma} \x}$.  Thus, by replacing $\x$ with $\sqrt{\Sigma} \x$, we recover the model in which $\Sigma = \Id$.  
 Further, we may substitute $\x'$ with $\sqrt{\Sigma} \x'$ in the $K$-Lasso to arrive at the $\sqrt{\Sigma} K$-Lasso:
\begin{equation}       
\text{minimize } \|\G x' - y\|_2 \text{ subject to } x' \in \sqrt{\Sigma} K
\end{equation}
where $\G$ is a matrix which contains $\g_i^\tran$ as its $i$-th row.  We have now completely reduced to the setup of Theorem \ref{thm: main}, with the caveat that we have substituted $\x,\x',$ and $K$ by $\sqrt{\Sigma} \x$, $\sqrt{\Sigma} \x'$, and $\sqrt{\Sigma} K$.  Apply the theorem to finish the proof of the corollary.
\end{proof}

\begin{remark}[Removing $\Sigma$ from the mean width]
If the covariance matrix $\Sigma$ is well conditioned, its effect on the 
error \eqref{eq: main with Sigma} can be easily evaluated using the inequality
\begin{equation}\label{eq: control mean width}
d(K, \Sigma) \leq \text{cond}(\Sigma) \cdot d(K).
\end{equation}
where $\text{cond}(\Sigma) = \|\Sigma\| \cdot \| \Sigma^{-1}\|$ denotes the condition number
and $d(K) = d(K, \Id)$ is the same as Theorem~\ref{thm: main}.
Before we prove this bound, let us mention that in some situations the
effect of $\Sigma$ is much smaller than it predicts -- for example, if $K$ is a subspace, 
then $d(K, \Sigma) = d(K)$.

To check \eqref{eq: control mean width}, note that for the tangent cone 
$D = D(K, \mu x)$ we have
\begin{equation}         \label{eq: w1 compare begin}
w_1(\sqrt{\Sigma}D) = \E \sup_{\x \in\sqrt{\Sigma} D \cap B_2} \ip{\x}{\g}
\le \|\sqrt{\Sigma}^{-1}\| \cdot \E \sup_{\x \in \sqrt{\Sigma}(D \cap B_2)} \ip{\g}{\x},
\end{equation}
where the inequality follows from the elementary containment 
$\sqrt{\Sigma} D \cap B_2 
\subset \|\sqrt{\Sigma}^{-1}\| \cdot \sqrt{\Sigma} (D \cap B_2)$.
A straightforward application of Slepians inequality \cite{V} 
then bounds the quantity in \eqref{eq: w1 compare begin} by
$\|\sqrt{\Sigma}^{-1}\| \cdot \|\sqrt{\Sigma}\| \cdot w_1(D)$.
Thus, we conclude \eqref{eq: control mean width}.
\end{remark}

\begin{remark}[Removing the assumption that $\|\sqrt{\Sigma} \x\|_2 = 1$]
The theory may be generalized to the case when $\|\Sigma \x\|_2 \neq 1$ with a simple rescaling argument.  Let $\delta = \|\sqrt{\Sigma} \x\|_2$ and let $\tilde{\x} := \x/\delta$.  Observe that
\[f(\ip{\a_i}{\x}) = f(\delta \ip{\a_i}{\tilde{\x}}) =: \tilde{f}(\ip{\a_i}{\tilde{\x}}).\]
Thus, the theorem applies to the estimation of $\tilde{x}$ with parameters
\[\mu := \E[\tilde{f}(g) \cdot g], \qquad \sigma^2 := \E (\tilde{f}(g) - \mu g)^2, \qquad \eta^2 := \E (\tilde{f}(g) - \mu g)^2 g^2.\]
\end{remark}
In some cases, one does not expect the tangent cone to have especially small mean width.  
As a motivating example, in the field of compressed sensing, it is standard to call $\x$ {\em compressible} if it belongs to a scaled $\ell_p$ ball for  $p \in (0,1)$, or if the ratio $\|\x\|_1/\|\x\|_2$ is small.  In this case, which contrasts with the case of exact sparsity, the tangent cone may have mean width comparable to the ambient dimension.  However, the set $K$ itself can still behave in a low-dimensional fashion.  Since $K$ is not necessarily a cone, and is not scale invariant, it is necessary to characterize dimension with a scaling parameter.  Fortunately, the local mean width accomplishes this task with $t$ as the scaling parameter, and  
$w_t(K - \mu \x)^2/t^2$ serving as a measure of the dimension at scale $t$.  

The next theorem considers a general signal structure.  

\begin{theorem}[Non-linear estimation without tangent cone structure]				\label{thm: main no cone}
  Suppose that $\a_i \sim \NN(0, \Id)$, 
  $\x \in S^{n-1}$, and that $\y$ follows the semi-parametric single index model of Section \ref{sec: model}.
  Assume that $\mu \x \in K$ where $K$ is convex,\footnote{More generally, the proof only requires that $K - \mu \x$ 
    be contained in a star shaped set.  This star shaped set can take the place of $K - \mu \x$ in the results of this theorem.}   
  and let $d_t (K) := w_t(K - \mu \x)^2 / t^2$.
  Then, the following holds with high probability. 
  For any $t > 0$ such that
  \[m \gtrsim  d_t(K),\]
   the non-linear estimator $\xhat$ which minimizes the $K$-Lasso \eqref{eq: estimation} satisfies 
  \begin{equation}         \label{eq: main no cone}
  \|\xhat - \mu \x\|_2 \lesssim  \frac{\sqrt{d_t(K)} \, \sigma + \eta}{\sqrt{m}} + t
  \end{equation}
  where $\mu$, $\eta$, and $\sigma$ are defined in \eqref{eq: parameters}.
\end{theorem}

Note that one may derive Theorem \ref{thm: main} by taking the limit as $t$ goes to zero in the above theorem.  However, in the proofs we will give a simpler and more straightforward route to the proof of Theorem \ref{thm: main}.

\begin{remark}[Non-trivial covariance matrix]
As above, this result can be generalized to the case when the covariance matrix of the rows is $\Sigma \neq \Id$.  One would just define $d_t(K, \Sigma)$ in a straightforward way similar to that in Corollary \ref{cor: unknown covariance}.
\end{remark}

\subsection{Key idea in the proof}
\label{ssec: key idea in the proof}
While it may be surprising that the $K$-Lasso is provably accurate even under the (non-linear) single-index model, it becomes much clearer when one observes that the expected loss, $\E ||\A \x' - \y||_2^2$, is minimized by $\mu \x$.  In other words, regardless of the form of the non-linearity, the expected squared error is minimized by a multiple of the original signal.  See Section \ref{sec: proofs} for a proof.

In fact, one may transform the single-index model into a scaled linear model with an unusual noise term.  Define an \textit{induced} noise vector $\z$ to satisfy
\[ \y = \A \mu \x + \z.\]
One may not expect $\z$ to play the role of noise, since it generally does not have zero mean, and is not independent of $\A$.  However, $z_i$ is uncorrelated with $\a_i$ (see Section \ref{sec: proofs}).  

We note that under this scaled linear model, one could use standard techniques to derive error bounds if $\z$ were deterministic, or independent of $\A$ \cite{oymak2013simple}, or if $\z$ were sub-Gaussian.  However, since we make quite mild assumptions in our single-index model, only implicitly assuming that the parameters $\mu$, $\sigma$, and $\eta$ are well-defined, this induced noise may have heavy tails and requires novel analysis.  Some of the tools for this analysis are available in the recent work \cite{pvy} by the current authors and Yudovina.  However, this earlier paper did not apply to the $K$-Lasso, and there were many technical details needed to extend these results.  In particular, the extra steps in the proof of Theorem \ref{thm: main no cone} are new ideas, as well as the method to give results with non-trivial covariance matrix.  We give a detailed comparison with this earlier work and others in the next section.

\section{Related literature}
\label{sec: related literature}
There is now a precise and comprehensive theory of signal reconstruction from {\em linear} observations, which takes into account signal structure.  While it is largely motivated by the quite modern area of \textit{compressed sensing} \cite{CSbook, CSbook2}, it is rooted in results developed in the older areas of \textit{geometric functional analysis} \cite{vershynin2011lectures, gordonEscapeMesh} and \textit{convex integral geometry} \cite{schneider2008stochastic}.  To leverage these tools, it is vital to assume that the measurement matrix $\A$ is random.  We give a brief overview of the results most closely aligned with this work.  The literature that we describe below takes $\A$ to be a matrix with independent Gaussian or sub-Gaussian entries.

In the noiseless case, signal reconstruction is possible as soon as the number of observations exceeds the manifold dimension \cite{eldar2012uniqueness}.  Even in the noisy case, there is a large pool of theory addressing signal reconstruction based on manifold dimension \cite{baraniuk2009random, wakin2010manifold, yap2011stable, eftekhari2013new}.  However, in the noisy case, it is necessary to make extra structural assumption of the set $K$ beyond assuming that it has small manifold dimension.  Otherwise, signal reconstruction based on a number of observations comparable to the manifold dimension can be unstable \cite{giryes2014effective}.  

The Gaussian mean width gives an alternative measure of dimension.  When it is applicable, it leads to simpler assumptions.  Indeed, as described above, the Gaussian mean width controls the conditioning of $\A$ when restricted to a cone, as proved in Gordon's escape through the mesh theorem.  Rudelson and Vershynin \cite{rudelson2008sparse} leveraged this result in the compressed sensing setup, showing that the signal could be reconstructed as long as the number of observations exceeded the squared Gaussian mean width of the tangent cone; Stojnic continued in this line of research \cite{stojnic2009various}.  Chandrasekaran et al.~\cite{chandrasekaran2012convex} extended this result to general convex bodies $K$.  Amelunxen et al.~\cite{ALMT} took a different route, synthesizing tools from conic integral geometry to give a precise phase transition for the number of observations needed to reconstruct $\x$.  There work is based on the \textit{statistical dimension}, which is roughly equivalent to the mean width, but has some extra convenient properties (see \cite{ALMT}).  This showed that previous results were tight.  A line of work by Thrampoulidis, Oymak, and Hassibi \cite{oymak2013simple, oymak2013squared, thrampoulidis2014simple} concentrated on the precise reconstruction error from noisy observations, and also considered unconstrained versions of the $K$-Lasso.  Our theoretical results in the non-linear case can be seen to mirror Theorem \cite[Theorem 1]{oymak2013simple} in the linear case.  We state a simplified version of this theorem, specialized to Gaussian noise (see the original theorem for a very careful treatment of constants).  

\begin{theorem}
  Suppose that $\a_i \sim \NN(0, \Id)$, 
  $\x \in S^{n-1}$, and that $\y$ follows the noisy linear model \eqref{eq: noisy linear model}.
  Assume that $\mu \x \in K$, and let 
  $d(K): = w_1(D(K, \x))^2$.
  Suppose that
  \[m \gtrsim d(K).\]
  Then, with high probability, the solution $\xhat$ of the $K$-Lasso \eqref{eq: estimation} satisfies 
  \[
  \|\xhat - \x\|_2 \lesssim \frac{\sqrt{d(K)} \, \sigma}{\sqrt{m}}.
  \]
\end{theorem}
Thus, one sees that our theorem \ref{thm: main}, when specialized to linear observations, recovers this modern theory up to an absolute constant.

\subsection{Prior work addressing non-linearity of the observations}
There are also numerous works, and fields of study, addressing non-linearity.
We describe the work that is most closely related to the present paper.

The semiparametric single-index model that we take in this paper is well studied in econometrics; see the monograph \cite{Horowitz}.  Most work in this area is asymptotic, although recent works have considered the finite case \cite{Hristache,Alquier,Tsybakov}.  However, we believe that this literature does not address, from a theoretical standpoint, the gains that can be made by utilizing a general low-dimensional structure.  See \cite[Section 6]{pvy} for a more thorough discussion of this literature.

In contrast, our work precisely characterizes the benefits from taking into account low-dimensional signal structure.  For example, consider the sparse signal structure assumed in compressed sensing, in which $\x$ contains at most $s$ non-zero entries.  The effective dimension is $O(s \log(n/s))$ which can be significantly smaller than the ambient dimension, $n$.  Thus, we show only $O(s \log(n/s))$ measurements are needed to estimate $\x$.  Specialized to the case of linear, noiseless measurements, our theory recovers the classic result that $\x$ may be exactly reconstructed from this number of measurements.  When non-linearity is present, the ``noise" induced by modeling non-linear measurements with linear measurements is reduced proportionally to $s \log(n/s)/m$.

The area of {\em 1-bit compressed sensing} \cite{1bitCSWebpage} concentrates on the case when the non-linearity is 1-bit quantization.  In other words, for $q \in \R$, $f(q) = \sign(q)$ or $f(q) = \sign(q + z)$ where $z$ is noise.  This has been a lively field of research for several years, in part due to a wide range of applicability in both signal processing problems and also statistical models in which the data is inherently binary.  The discrete nature of this problem has led to new challenges that were not inherent in unquantized compressed sensing.  Indeed, even the method of reconstruction of the signal has posed a challenge, and some of the proposed methods, such as the approach of \cite{pvy} require knowledge of the covariance of the rows to be accurate.  We believe our paper provides the first analysis of the $K$-Lasso for this problem, and the first theoretical result which allows non-trivial covariance of the rows of $\A$.  In the next section, we specialize our work to the 1-bit compressed sensing model.

While there are numerous other publications which relate to various forms of non-linearity and low-dimensionality, there are three papers which we believe are most closely related to our results \cite{pvy, negahban2012unified, lecue2013learning}.  All three papers address general low-dimensional signal set $K$ combined with general non-linearity.  Our current result builds on the work in \cite{pvy}, which considers a very similar model.  There are two significant extensions that we make beyond this work.  First, our results are tighter in the sense that when specialized to the linear model, they match modern theory which is developed specifically for the linear model (see above).  This is only true in \cite{pvy} when the noise is larger than the signal.  Further, as discussed above, the method espoused in \cite{pvy} is not the $K$-Lasso, and requires knowledge of $\Sigma$ to be effective.  

The other two related works \cite{negahban2012unified, lecue2013learning} give a very general framework, which does not focus on the $K$-Lasso, but can be specialized to this recovery method.  
We believe that using the framework of \cite{negahban2012unified}, a theorem similar to our Theorem~\ref{thm: main} could be derived.  A key statistical idea, which is put rigorously in \cite{negahban2012unified}, is that the solution to the $K$-Lasso is a good estimate of the minimizer of the expected loss.  In other words, misspecification of the model is tolerable provided that the true signal minimizes the expected loss.  See \cite[Theorem 1]{wainwright2014structured} for a simplified version of this result. As we noted in Section \ref{ssec: key idea in the proof}, $\mu \x$ is indeed the minimizer of the expected loss---this is the first step in our proofs, and could be used as a first step to derive error bounds from the framework of \cite{negahban2012unified}.    However, the results of \cite{negahban2012unified} are general enough that  such a derivation is non-trivial.  Furthermore, we do not require \textit{restricted strong convexity} in our Theorem \ref{thm: main no cone} or \textit{decomposability} in any of our theorems, which are two strong requirements of \cite{negahban2012unified}.  Similarly, by observing that $\mu \x$ minimizes expected loss, the results of \cite{lecue2013learning} could be specialized to the $K$-Lasso.  This would give a result similar to our Theorem \ref{thm: main no cone}.  However, our result expands upon this in two ways: 1)  In \cite{lecue2013learning} it is assumed that $y_i$ is sub-Gaussian, whereas we make almost no assumption on $y_i$---roughly, it only needs a bounded second moment; 2) In contrast to \cite{lecue2013learning}, our theory takes advantage of local structure of $K$ around $\mu \x$, thus allowing, for example, the consideration of tangent cones.  By doing this, our theory re-creates classical compressed sensing results as a special case, for example.

Finally, we would like to point to the new work \cite{thrampoulidis2015lasso} which considers the unconstrained version of the $K$-Lasso. By considering the asymptotic regime and adopting a stochastic model for signals $\x$, the authors of \cite{thrampoulidis2015lasso} were able to give a precise treatment of constants involved in the error bounds.

\section{Specialization to 1-bit compressed sensing}
\label{sec: 1-bit}
As discussed above, the simplest 1-bit compressed sensing model takes the following form:  For $q \in \R$, $f(q) = \sign(q)$, i.e., we just observe the sign of the linear observations. 
Let $K$ be a scaling of the $\ell_1$ ball and $\x$ is assumed to be $s$-sparse, i.e., to contain only $s$ non-zero entries.  This latter requirement implies that the tangent cone has small mean width.  Indeed, as can be seen from \cite{chandrasekaran2012convex} for instance, for the appropriate scaling of $K$, one has 
\[d(K) = w_1(K - \mu \x)^2 \lesssim s \log(n/s).\] 
A straightforward calculation shows that 
\[\mu = \sqrt{\frac{2}{\pi}}, \quad \sigma^2 = 1 - \frac{2}{\pi}, \quad \eta^2 = 1 - \frac{2}{\pi}.\]

Thus, Theorem \ref{thm: main} states that as long as $m = O(s \log(n/s))$ observations are observed, the $K$-Lasso gives accuracy 
\begin{equation}         \label{eq: 1-bit recovery}
\|\xhat - {\textstyle \sqrt{\frac{2}{\pi}}} \, \x\|_2 \lesssim \sqrt{ \frac{s \log(n/s)}{m} }.
\end{equation}
Moreover, this bound holds for observations with general covariance structure. 
Indeed, Corollary \ref{cor: main} combined with \eqref{eq: control mean width} imply that 
\eqref{eq: 1-bit recovery} remains true as long as $\Sigma$ is reasonably well conditioned.

This yields the following surprising conclusion:
\begin{quote}
\textit{Even for highly non-linear observations, such as 1-bit quantization, the $K$-Lasso is quite 
accurate as long as the number of observations significantly exceeds the effective dimension of the signal.}
\end{quote}
\section{Proof of main results}
\label{sec: proofs}

We begin by setting
\[\z := \y - \A \mu \x.\]
While $\z$ is not independent of $\A$ or $\x$, and generally does not have mean 0, it will nevertheless play the role of noise.  As shown in \cite{pvy}, $\z$ satisfies
\begin{equation}\label{eq: z satisfies}
\E \A^\tran \z = 0.
\end{equation}
We repeat the derivation here to keep the paper self contained.  It suffices to show that for any 
$\v \in S^{n-1}$, $\E \v^\tran \A^\tran \z = 0$, which in turn would follow from 
\[
\E y_i \ip{\a_i}{\v} - \E \mu \ip{\a_i}{\v} \ip{\a_i}{\x} = 0.
\]
Since the covariance of $\a_i$ is identity, the second term is equal to $\mu \ip{\x}{\v}$.  
To calculate the first term, note that $g_i := \ip{\a_i}{\x}$ has distribution $\NN(0,1)$. 
Then make the Gaussian decomposition $\ip{\a_i}{\v} = \ip{\x}{\v} g_i + g_i^\perp$ where $g_i^\perp$ 
is independent of $g_i$. By independence, the first term above is equal to
\[\E y_i \ip{\a_i}{\v} = \E f(g_i) [\ip{\x}{\v} g_i + g_i^\perp] = \ip{\x}{\v} \E f(g_i) g_i = \mu \ip{\x}{\v}\]
where the first equality follows from our model assumption \eqref{eq: model} that 
$y_i = f(g_i)$, and the last equality follows by definition of $\mu$ in \eqref{eq: parameters}.
This completes the derivation of \eqref{eq: z satisfies}.

Now let $\widehat{\x}$ be the solution of the $K$-Lasso \eqref{eq: estimation},
that is the minimizer of the loss function $\| \A \x' - \y \|_2$ on $K$.
We may replace this loss function by
\[L(\x') := \frac{1}{m} \left(\| \A \x' - \y \|_2^2 - \| \A \mu \x - \y \|_2^2\right) \]
without affecting the minimizer $\widehat{\x}$.  Indeed, $\mu \x$ is a fixed scalar multiple of a fixed signal, 
and thus we have only squared the loss function, subtracted a constant and multiplied by $1/m$.
Now, the new loss function is very well-behaved in expectation.  

\begin{lemma}[Expected loss]
\label{lem: expected loss}
\[\E L(\x') = \| \x' - \mu \x\|_2^2.\]
\end{lemma}

\begin{proof}
Expanding $L(\x')$, we can express it more conveniently as
\begin{equation}\label{eq: L expression}
L(\x') = \frac{1}{m} \| \A \h \|^2 - \frac{2}{m} \ip{\h}{\A^\tran \z}
\quad \text{where} \quad \h := \x' - \mu\x.
\end{equation}
The second term has zero mean according to \eqref{eq: z satisfies}. 
Since the covariance matrix of $\a_i$ is identity, the first term is $\|\h\|_2^2$ in expectation, as desired.
\end{proof}

Lemma \ref{lem: expected loss} implies that $\mu \x$ minimizes the {\em expected} loss.  
In order to prove the main theorem, we need to control the deviation from expectation 
of the two terms in the loss function \eqref{eq: L expression}.

First, we lower bound the ratio of $\frac{1}{m} \|\A \h\|_2^2$ to its expectation value of $\|\h\|_2^2$. 
This can be done by applying the classical result from the work of Gordon \cite{gordonEscapeMesh}. 

\begin{theorem}[Escape through the mesh] \label{thm: gordon}
Let $D \subset \R^n$ be a cone. 
Then
\begin{equation} \label{eq: gordon}
\inf_{\v \in D \cap S^{n-1}} \| \A \v \|_2 \geq \sqrt{m - 1} - w_1(D) - r
\end{equation}
with probability at least $1 - e^{-r^2/2}$. 
\end{theorem}

Next, we control the size of $\ip{\h}{\A^\tran\z}$.

\begin{lemma}  \label{lem: control cross term}
Let $D \subset t B_2^n$, and let $\z := \y - \A \mu \x$ as before.  Then 
\begin{equation} \label{eq: control cross term}
\E \sup_{\v \in D} \ip{\v}{\A^\tran\z} \leq 
C \left( w(D) \sigma + t \eta \right) \sqrt{m}.
\end{equation}
\end{lemma}
Here and in the rest of the argument, $C, c$ refer to numerical constants; their values may differ from instance to instance. Before proving Lemma~\ref{lem: control cross term}, we pause to show how the lemma and Theorem \ref{thm: gordon} imply our main result.

\begin{proof}[Proof of Theorem \ref{thm: main}]
For convenience, let us denote the spherical part of the tangent cone by 
$D=D(K,\mu \x) \cap S^{n-1}$.
We begin by recording two events which occur with high probability.  
First, under the assumptions of our main Theorem \ref{thm: main}, 
the escape through the mesh Theorem~\ref{thm: gordon} implies that the following event holds with probability 
at least $0.995$:
\[
\text{Event 1:} \quad \inf_{\v \in D} \frac{1}{\sqrt{m}} \| \A \v \|_2 \geq c.
\]
Second, Markov's inequality combined with Lemma \ref{lem: control cross term} implies that the following event holds with probability at least $0.995$:
 \[
\text{Event 2:} \quad  \sup_{\v \in D} \ip{\v}{\A^\tran\z} \leq 
C \left( w(D) \sigma + \eta \right) \sqrt{m}.
\]
By the union bound, both events hold together with probability at least $0.99$.  (We note in passing that the probability of success, and also the constant $C$ in the bound of Event 2 could be sharpened using concentration inequalities.  However, this would not change our final presentation.)

We now show how to bound the error vector $\h:= \widehat{\x} - \mu \x$ in the intersection of these events. 
Since $\widehat{\x}$ minimizes the loss, we have
\[L(\widehat{\x}) \leq L(\mu \x) = 0.\]
Combine this with Equation \eqref{eq: L expression} to give
\begin{equation} \label{eq: h satisfies}
\frac{1}{m} \| \A \h \|^2 \leq  \frac{2}{m} \ip{\h}{\A^\tran\z}.
\end{equation}
 
On the other hand, $\h$ belongs to the tangent cone $D(K, \mu\x)$, so $\v:= \h/\|\h\|_2$ belongs to 
its spherical part $D = D(K,\mu\x) \cap S^{n-1}$. Then, by Events 1 and 2, we have
\[
\frac{1}{m} \|\A \h\|_2^2 \ge c \|\h\|_2^2 
\quad \text{and} \quad
\ip{\h}{\A^\tran \z} \leq \|\h\|_2 \cdot C \left( w(D) \sigma + \eta \right) \sqrt{m}.
\]
Combining these two inequalities with \eqref{eq: h satisfies}, we obtain
\[
c \|\h\|_2^2 
\le \frac{2}{m} \cdot \|\h\|_2 \cdot C \left( w(D) \sigma + \eta \right) \sqrt{m}.
\]
Simplifying this bound, we complete the proof.
\end{proof}

We now prove Lemma \ref{lem: control cross term}.
\begin{proof}[Proof of Lemma \ref{lem: control cross term}]
This proof has similar steps to the proof of Theorem 1.3 in \cite{pvy}.  
We begin with a projection argument to (mostly) decouple $\z$ from $\A$.  
Let $\P := \x \x^\tran$ be the orthogonal projection onto the span of $\x$ 
and let $\P^{\perp} := \Id -  \x \x^\tran$ be the projection onto the orthogonal complement.  
Then, convexity of the functional $\|\u\|_{D^\circ} := \sup_{\v \in D} \ip{\v}{\u}$ 
leads to the following decomposition:
\[
\E \|\A^\tran\z\|_{D^\circ} \leq \E \|\P^{\perp} \A^\tran\z\|_{D^\circ} + \E \|\P \A^\tran\z\|_{D^\circ} =: I + II.
\]

We first control $I$.
Note that, since $\A$ is Gaussian, $\P^{\perp} \A^\tran$ is independent from $\P \A^\tran$. 
It follows that $\P^{\perp} \A^\tran$ is also independent of $\z$. 
Indeed, to obtain the latter conclusion, simply note that the columns of $\P \A^\tran$ are $\ip{\a_i}{\x}\x$, 
and the coordinates of $\z$ are 
\begin{equation}         \label{eq: coordinates of z}
z_i = f(\ip{\a_i}{\x}) - \mu \ip{\a_i}{\x}.
\end{equation}
Therefore, $\P^{\perp} \A^\tran\z$ is distributed identically with $\P^{\perp} \tilde{\A}^\tran\z$, 
where $\tilde{\A}$ is an independent copy of $\A$ (independent also of $\z$). Thus
\[ I = \E \|\P^{\perp} \A^\tran\z\|_{D^\circ} = \E \|\P^{\perp} \tilde{\A}^\tran\z\|_{D^\circ} = \E \|(\P^{\perp} \tilde{\A}^\tran + \E [\P \tilde{\A}^\tran])\z\|_{D^\circ}.\]
Now, by Jensen's inequality, the last quantity is bounded by
\[\E \|(\P^{\perp} \tilde{\A}^\tran + \P \tilde{\A}^\tran)\z\|_{D^\circ} = \E \|\tilde{\A}^\tran\z\|_{D^\circ}.\]
Now condition on $\z$.  Then $\tilde{\A}^\tran\z$ has distribution $\|\z\|_2 \cdot \NN(0, \Id)$.  Thus
\[
I \leq \E \|\tilde{\A}^\tran\z\|_{D^\circ} = \E \|\z\|_2 \cdot w(D) \leq \sqrt{\E \|\z\|_2^2} \cdot w(D) 
= \sqrt{m} \, \sigma \cdot w(D).
\]
Here in the first equality we used the definition of $w(D)$; in the last equality, 
we recall \eqref{eq: coordinates of z} and definition of $\sigma$ from \eqref{eq: parameters}.

We now control $II$.  Note that 
\[
\P \A^\tran\z 
= \sum_{i=1}^m z_i \ip{\a_i}{\x}\x
= \sum_{i=1}^m \xi_i \cdot \x
\] 
where $\xi_i := z_i \ip{\a_i}{\x} =  \big[ f(\ip{\a_i}{\x}) - \mu \ip{\a_i}{\x} \big] \ip{\a_i}{\x}$.
Thus, 
\[
II \le \| \x \|_{D^\circ} \cdot \E \Big| \sum_{i = 1}^m \xi_i \Big|.
\]
Since $D \subset t B_2^n$, we have $\| \x \|_{D^\circ} \leq t$. Substituting this, we obtain
\[
II \leq t \E \Big| \sum_{i = 1}^m \xi_i \Big| 
\leq t \sqrt{\sum_{i = 1}^m \E \xi_i^2} = t \sqrt{m \E \xi_1^2} = t \sqrt{m} \cdot \eta
\]
where the last equality follows by definition of $\eta$ from \eqref{eq: parameters}.
The proof is complete.
\end{proof}

\subsection{Proof of Theorem \ref{thm: main no cone}}
When the error vector $\h= \xhat - \mu \x$ is not known to belong to a cone, 
but rather a general set, it can no longer be guaranteed that $\h$ is not in the null space of $\A$
(which was true for cones via Gordon's Theorem~\ref{thm: gordon}.)
Nevertheless, such bad behaviour generally only occurs at tiny scales, and at large scales $\A$ 
may be quite well conditioned even on general sets. 
This idea is made rigorous in the following lemma, which is known in the geometric functional analysis
community even in more generality, see \cite{schechtman2006two, klartag2005empirical, mendelson2007reconstruction, mendelson2014learning, tropp2014convex}. 
For the sake of the reader, we will include 
a proof below. 

\begin{lemma} \label{lem: large scale conditioning}
  Let $K \subset \R^n$ be a star shaped set.\footnote{$K$ is a star shaped set if it
    satisfies $\lambda K \subset K$ for any $0 \leq \lambda \leq 1$.}
  Let $t > 0$ and suppose that $m \gtrsim w_t(K)^2/t^2$.  
  Then, with probability at least $1 - 2\exp(-m/8)$, the following holds for all $\v \in K$ satisfying $\|\v\|_2 \geq t$:
  \[\|\A \v\|_2 \geq c \sqrt{m} \|\v\|_2.\]
\end{lemma}

Before proving this lemma, let us combine it with Lemma \ref{lem: control cross term} to prove the second main result.  

\begin{proof}[Proof of Theorem \ref{thm: main no cone}]
For convenience, let us denote $K_x := K - \mu \x$.
As before, we begin by considering two good events, whose intersection holds with probability at least $0.99$, 
based on Lemma \ref{lem: large scale conditioning} and Lemma \ref{lem: control cross term}.  
\begin{align*}
&\text{Event 1:} \quad \inf_{\v \in K_x \cap t B_2^c} \frac{1}{\sqrt{m}} \frac{\| \A \v \|_2}{\|\v\|_2} \geq c.\\\vspace{3mm}
& \text{Event 2:} \quad  \sup_{\v \in K_x \cap t B_2} \ip{\v}{\A^\tran\z} \leq C \left( w_t(K_x) \sigma + t \eta \right) \sqrt{m}.
 \end{align*}

We now show how to bound the error vector $\h:= \widehat{\x} - \mu \x$ in the intersection of these events. 
As in the proof of Theorem~\ref{thm: main}, the fact that $\xhat$ minimizes the loss implies that 
\begin{equation}							\label{eq: Ah small}
\frac{1}{m} \| \A \h \|_2^2 \leq  \frac{2}{m} \ip{\h}{\A^\tran\z}.
\end{equation}
We can assume that $\|\h\| \geq t$, since in the opposite case 
the error bound of Theorem~\ref{thm: main no cone} holds trivially.  
Since $\h \in K_x$, 
the inequality of Event 1 followed by \eqref{eq: Ah small} gives
\begin{equation}							\label{eq: h via ip}
c^2 \|\h\|_2^2 \leq \frac{2}{m} \ip{\h}{\A^\tran\z}.
\end{equation}
We would like to apply the inequality of Event 2, but cannot do this directly because $\|\h\|_2$ is not bounded above by $t$.  
Fortunately, since $K_x$ is convex and contains the origin, $K_x$ is star shaped.  Using this fact, we may massage our bound into the form of Event 2 via a monotonicity argument.  
 
Divide both sides of \eqref{eq: h via ip} by $\delta := \|\h\|_2$. This gives
\begin{equation}							\label{eq: delta}
c^2 \delta \leq \frac{2}{m} \d^{-1} \ip{\h}{\A^\tran\z} 
\le \frac{2}{m} \sup_{\u \in \d^{-1} K_x \cap B_2} \ip{\u}{\A^\tran\z} =: f(\delta),
\end{equation}
where in the second inequality we set $\u = \d^{-1} \h$ and used that $\h \in K_x$ again.
Now, since $K_x$ is star shaped, $f(\delta)$ is a monotonically decreasing function.
Thus, by assumption $\delta \ge t$, we may replace $\delta$ by $t$ in our bound, giving
\[
c^2 \|\h\|_2 \leq f(t) 
= \frac{2}{mt} \sup_{\v \in K_x \cap t B_2} \ip{\v}{\A^\tran\z}.
\]
The proof is completed by applying the inequality of Event 2.
\end{proof}

It remains to prove Lemma~\ref{lem: large scale conditioning}.

\begin{proof}
We begin with the following simple comparison, which follows from the Cauchy-Schwartz inequality
for all $\v \in \R^n$:
\begin{equation}				\label{eq: comparison}
\|\A\v\|_2 \geq \frac{\|\A\v\|_1}{\sqrt{m}}.
\end{equation} 
Furthermore, since $K$ is star shaped, we have
\begin{equation}
\label{eq: star-shaped satisfies}
\inf_{\v \in K \cap t B_2^c} \frac{\|\A \v\|_1}{\|\v\|_2} = \inf_{\u \in K \cap t S^{n-1}} \frac{\|\A \u\|_1}{t}.
\end{equation}
(Indeed, $\u = t\v/\|\v\|_2$ lies in $K$ since $t/\|\v\| \le 1$ and $K$ is star shaped.)

Next, we will control $\|\A\u\|_1$ with an application of the following uniform deviation inequality, 
which we proved in \cite{PV_DCG}.

\begin{lemma}[Uniform deviation for the $\ell_1$ norm]  
  Let $K \subset \R^n$ and let $r, t > 0$. 
  Then, with probability at least $1 - 2\exp(-mr^2/t^2)$, 
  the following holds for all $\u \in K$ satisfying $\|\u\|_2 \le t$: 
  \[
  \Big| \frac{1}{m} \|\A \u\|_1 - \sqrt{\frac{2}{\pi}} \, t  \Big| 
  \le \frac{4 w_t(K)}{\sqrt{m}} + r. 
  \]
\end{lemma}
Choosing $r = t/2$ in this lemma, we conclude that with probability at least $1 - \exp(-m/8)$, one has 
\begin{equation}			\label{eq: Av1}
\inf_{\u \in K \cap t S^{n-1}} \frac{1}{m} \|\A \u\|_1 \geq c t 
\quad \text{where} \quad c = \sqrt{\frac{2}{\pi}} - \frac{1}{2} - \frac{4 w_t(K)}{t \sqrt{m}}.
\end{equation}
Recalling the assumption of Lemma~\ref{lem: large scale conditioning} that $m \gtrsim w_t(K)^2/t^2$, 
we see that $c$ is bounded below by a positive absolute constant.  
In this case, we can substitute the bound into \eqref{eq: star-shaped satisfies} to obtain
$$
\inf_{\v \in K \cap t B_2^c} \frac{\|\A \v\|_1}{\|\v\|_2} \ge cm.
$$
We finish the proof by an application of inequality \eqref{eq: comparison}.
\end{proof}

\section{discussion}

We have analyzed the $K$-Lasso for signal reconstruction from the semiparametric single-index model.  
We showed that the $K$-Lasso solution under the non-linear model $y_i = f(\ip{a_i}{x})$ 
behaves roughly like the $K$-Lasso solution under the noisy linear model $y_i = \mu x + \sigma z_i$ with
$z_i \sim N(0,1)$, where $\mu = \mu(f)$ and $\sigma = \sigma(f)$ have simple expressions; the error of the $K$-Lasso is controlled by the local mean width of $K$.  We hope this theoretical result may aid researchers who use the $K$-Lasso in situations when the response may not be linear.  See \cite{chretien2015using} for one such implementation.  

We have made some idealized assumptions in this paper thus allowing theoretical results that are simple to state and understand.  There are many future directions of research both of theoretical and practical interest, particularly in softening assumptions, which we describe below.

We considered a Gaussian design matrix, $A$, and this allowed for a clean theoretical result.  It is of interest to determine whether these results have some universality properties.  Can the same kind of accuracy be expected for random non-Gaussian matrices?  Under the linear model, universality results have been shown in the compressed sensing literature \cite{donoho2009observed}, that is, theoretical performance based on a Gaussian matrix is shown to empirically match the performance for many other kinds of matrices.  However, there is an extra wrinkle under the single-index model: a universality result is impossible when $x$ is extremely sparse \cite{ALPV}.  When $A$ has independent sub-Gaussian entries, we conjecture that the results of our paper should still hold, although with an extra error term that becomes large when $x$ is very sparse, and shrinks towards zero if $x$ is spread out.  It is of interest to iron out this theory and also to determine, both theoretically and empirically, how far these results may extend towards general design matrices.

Another direction of interest is robustness of the $K$-Lasso to model inaccuracies.  
Will the $K$-Lasso solution remain accurate if the single-index model is only approximately true, or if $\mu x$ does not quite reside in $K$? 

Finally, these results lead to new opportunities in signal processing problems in which the scientist has some control over the non-linearity $f$, e.g., for quantization (see \cite{thrampoulidis2014simple}).  In that case, the explicit expressions for $\mu(f)$ and $\sigma(f)$ may be tuned to optimize the error.  It is of interest to identify other such problems, aside from quantization, that can benefit from this.

\bibliographystyle{plain}
\bibliography{pv-nonlinear-Lasso}
\end{document}